\def\Title{Quantum Disturbance without State Change:\\
Soundness and Locality of Disturbance Measures}
\def\Author{Masanao Ozawa}
\documentclass[12pt]{article}    
\usepackage{cite}
\usepackage{url}
\usepackage{times}
\usepackage{amsmath}
\usepackage{amsfonts}
\usepackage{amssymb}
\usepackage{graphicx}
\usepackage{amsthm} 
\newtheorem{Theorem}{Theorem}
\usepackage{color} 
\usepackage[normalem]{ulem}
\usepackage{soul}

\title{\Large\bf \Title} 
\author{\Author$^{1,2,}$\thanks{\em\small  E-mail:  ozawa@is.nagoya-u.ac.jp.}
\\
\\
{\em\small 
{}$^1$ Center for Mathematical Science and Artificial Intelligence, 
Academy of Emerging Sciences,}\\
{\em\small 
Chubu University, 1200 Matsumoto-cho, Kasugai 487-8501, Japan}\\
{\em\small 
{}$^2$ Graduate School of Informatics, Nagoya University, 
Chikusa-ku, Nagoya 464-8601, Japan}\\
}
\date{}

\begin{document} 

\maketitle 
\thispagestyle{headings}
\noindent
\begin{abstract}
It is often supposed that a quantum system is not disturbed without state change.
In a recent debate, this assumption is used to claim that the operator-based
disturbance measure, a broadly used disturbance measure, has an unphysical property.
Here, we show that a quantum system possibly incurs an operationally detectable 
disturbance without state change to rebut the claim. 
Moreover, we establish the reliability, formulated as soundness and locality, 
of the operator-based disturbance measure, which, we show, quantifies the disturbance 
on an observable that manifests in the time-like correlation even in the case 
where its probability distribution does not change.  
\end{abstract}

\section{Introduction}
Heisenberg's error-disturbance relation (EDR)
\begin{equation}\label{eq:Hei27}
\varepsilon(A)\eta(B)\ge\frac{1}{2}|{\langle {[A,B]} \rangle}|
\end{equation}
for the mean error $\varepsilon(A)$ of a measurement of an observable $A$ in any state 
and the mean disturbance $\eta(B)$ caused on an observable $B$,
originally introduced by the $\gamma$-ray microscope thought experiment \cite{Hei27},  
has been commonly believed as a dynamical aspect of Heisenberg's uncertainty principle, 
which is formally represented by a rigorously proven relation
\begin{equation}\label{eq:Rob29}
\sigma(A)\sigma(B)\ge\frac{1}{2}|{\langle {[A,B]} \rangle}|
\end{equation}
for the indeterminacies, defined as the standard deviations $\sigma(A),\sigma(B)$, 
of arbitrary observables $A,B$ 
in any state \cite{Hei27,Ken27,Rob29}. 
There have been longstanding research efforts to prove Heisenberg's EDR 
\cite{AK65,YH86, AG88,Ish91,91QU}, while the universal validity has not been reached.
Instead, a recent study  \cite{03UVR,04URN} revealed a
universally valid form of EDR
\begin{equation}\label{eq:UVEDR}
\varepsilon(A)\eta(B)+\varepsilon(A)\sigma(B)+\sigma(A)\eta(B)\ge\frac{1}{2}|{\langle {[A,B]} \rangle}|,
\end{equation}
where $\sigma(A)$ and $\sigma(B)$ are the standard deviations just before the
measurement, 
and made Heisenberg's EDR testable  \cite{04URN,LW10}
 to observe its violations, confirming the new relation as well
 \cite{12EDU,RDMHSS12,13VHE}.
Subsequently, stronger EDRs were derived \cite{Bra13,Bra14,14EDR,19A1}, and
 confirmed experimentally
\cite{13EVR,WHPWP13,14A1,RBBFBW14,16A3,Liu19}.
 
In order to define the error $\varepsilon(A)$ and disturbance $\eta(B)$ in Eq.~(\ref{eq:UVEDR}),
we suppose that the measurement $\mathbf{M}$ of $A$ is described by an interaction from time 
$t=0$ to $t=\tau$ between the system $\mathbf{S}$ in a state $|{\psi}\rangle$ and the probe $\mathbf{P}$ 
 prepared in a fixed state $|{\xi}\rangle$, and that the outcome of the measurement is obtained by the measurement 
 of  the meter observable $M$ in the probe $\mathbf{P}$ at time $t=\tau$.
 \footnote{
 Note that this general description of a measuring process, also called an indirect measurement model \cite{04URN}, 
 is introduced and proved in Ref.~\cite{84QC}
 to be equivalent to the most general description using a completely positive instrument, 
 or a so-called quantum instrument, which is a reformulation of the Davies-Lewis instrument \cite{DL70}
 with the additional requirement of complete positivity. 
 }
 
In the Heisenberg picture, we shall write $X(0)=X\otimes I$, $X(\tau)=U^{\dagger}X(0)U$,
$Y(0)=I\otimes Y$, $Y(\tau)=U^{\dagger}Y(0)U$ for observables $X$ in $\mathbf{S}$
 and $Y$ in $\mathbf{P}$,
where  $U$ is the unitary evolution operator for $\mathbf{S}+\mathbf{P}$ from $t=0$ to $t=\tau$.
The error $\varepsilon(A)=
\varepsilon_{O}(A,\mathbf{M},|{\psi}\rangle)$ and
disturbance  $\eta(B)=
\eta_{O}(B,\mathbf{M},|{\psi}\rangle)$ in Eq.~(\ref{eq:UVEDR}) are defined 
by
\begin{align}{
\varepsilon_{O}(A,\mathbf{M},|{\psi}\rangle)&=\langle{\psi,\xi|[M(\tau)-A(0)]^2|\psi,\xi}\rangle^{1/2},\\
\eta_{O}(B,\mathbf{M},|{\psi}\rangle)&=\langle{\psi,\xi|[B(\tau)-B(0)]^2|\psi,\xi}\rangle^{1/2}.\label{eq:O-dist}
}\end{align}
See Ref.~\cite{04URN} for details.
We call $\varepsilon_{O}$ and $\eta_{O}$ as the {\em operator-based 
error measure} 
and the {\em operator-based disturbance measure}.
We shall write $\varepsilon_{O}(A)=\varepsilon_{O}(A,\mathbf{M},|{\psi}\rangle)$ and 
$\eta_{O}(B)=\eta_{O}(B,\mathbf{M},|{\psi}\rangle)$ when no confusion may occur.

In the previous work \cite{19A1}, we have investigated the properties of the operator-based error measure
(called therein as noise-operator-based quantum root-mean-square error)
$\varepsilon_{O}$, and we have introduced  its completion $\overline{\varepsilon}$,
the locally uniform quantum root-mean-square error,
and subsequently we have experimentally tested \cite{21A7} the completeness of 
$\varepsilon_{O}$ and $\overline{\varepsilon}$ to show how hidden error in  $\varepsilon_{O}$ 
manifests in the defining procedure of $\overline{\varepsilon}$.

In the present work, we focus on the properties, soundness and locality, of the operator-based 
disturbance measure $\eta_{O}$, where soundness generally requires a disturbance measure 
to assign the value 0 to ``non-disturbing'' measurements, and locality generally requires 
a disturbance measure to assign the value 0 to ``non-disturbing'' local measurements.

We say that a measurement is {\em distributionally non-disturbing} to an
observable $B$ in the system state $|{\psi}\rangle$ 
if $B(0)$ and $B(\tau)$  have identical probability distributions 
in the initial state $|{\psi,\xi}\rangle$.
Korzekwa, Jennings, and Rudolph \cite{KJR14} criticized the use of the 
operator-based disturbance measure, based on the following requirement 
for disturbance measures.  

\sloppy
{\em  Distributional requirement (DR) for disturbance measures.}
{Any disturbance measure should assign the value 0 to distributionally
non-disturbing measurements.}\footnote{
Note that KJR \cite{KJR14} called distributionally non-disturbing
measurements as ``operationally non-disturbing measurements''; 
see Eq.~(2) in KJR \cite{KJR14}. }

KJR \cite{KJR14} called the DR  ``the commonly accepted and operationally 
motivated requirement that all physically meaningful notions of disturbance should
satisfy".  They claimed that the operator-based disturbance measure 
does not satisfy the DR and has even an `unphysical' property,
since it takes a positive value for a measurement that does not change the state at all. 
Further, they concluded that state-dependent formulations of EDRs 
are not tenable.

In this paper, we examine the validity of the DR.
For this purpose, we consider a more fundamental principle in quantum mechanics, 
the correspondence principle, stating that if the classical description is available, 
quantized concepts should be consistent with the classical description.  
We argue that the DR violates the correspondence principle. 
We generally show that even if the measurement does not change the state, 
the disturbance is operationally detectable as long as the operator-based 
disturbance measure takes a positive value.  
Thus, the claims made by KJR are groundless.
The DR requires that disturbance measures only count the change of the
probability distribution in time, but according to the correspondence principle, 
valid disturbance measures should also count the change of the observable 
that manifest in the time-like correlation, as the operator-based disturbance 
measure does.

Moreover, we show that the DR violates another fundamental requirement
for no-signaling under local operations, called the locality requirement.  
Subsequently, we show that disturbance measures satisfying the DR 
cannot be used to demonstrate the security of quantum 
cryptography, because they do not properly describe the disturbance caused 
by the eavesdropper.
In contrast, we show that the operator-based disturbance measure satisfies the
correspondence principle and the locality requirement.
Based on those arguments, we shall conclude that state-dependent formulations of EDRs 
based on the operator-based disturbance measure reliably represent 
the originally motivated dynamical aspect of Heisenberg's uncertainty principle.

\section{Correspondence principle}  
The correspondence principle generally states that quantum theory should be consistent
with classical theories in the case where the classical descriptions are also available.\footnote{
``The term [the correspondence principle] codifies the idea that a new theory should reproduce 
under some conditions the results of older 
well-established theories in those domains where the old theories work [Wikipedia \url{https://en.wikipedia.org/wiki/
Correspondence_principle} (August 1, 2022)]''.
}
In fact, it is a common practice to apply classical 
descriptions to commuting observables through their joint probability 
distributions. 
In Ref.~\cite{19A1} we consider the correspondence principle as a requirement 
for error measures.  Here, we extend the consideration to disturbance measures.

It is well-known that any commuting observables $X,Y$ have their joint probability 
distribution in any state.
Here, for a given state $|{\Psi}\rangle$, the {\em joint probability distribution (JPD)} 
of any two observables $X,Y$ 
is defined as a 2-dimensional probability distribution $\mu(u,v)$ satisfying 
\begin{align}\label{eq:JPD-CP}
\langle{\Psi|f(X,Y)|\Psi}\rangle=\sum_{u,v}\, f(u,v)\mu(u,v)
\end{align}
for every (non-commutative) polynomial $f(X,Y)$ of $X$ and $Y$. 
In general, two observables $X,Y$ have their JPD in a state $|{\Psi}\rangle$ 
if and only if they commute in $|{\Psi}\rangle$ in the sense that
\begin{align}
P^{X}(u)P^{Y}(v)|\Psi\rangle=P^{Y}(v)P^{X}(u)|\Psi\rangle,
\end{align}
where $P^{X}(u)$ and $P^{Y}(v)$ are the spectral projections of $X$ and $Y$
 (Ref.~\cite{19A1}, Theorem 1).
In this case, the JPD $\mu$ is uniquely determined by
\begin{align}{
\mu(u,v)=\langle{\Psi|P^{X}(u)P^{Y}(v)|\Psi}\rangle.
}\end{align}
The JPD $\mu$ determines the {\em (classical) root-mean-square deviation} 
$\delta_G(\mu)$ between the classical random variables ${\bf u}=u$ and ${\bf v}=v$,
the notion originally introduced by Gauss \cite{Gau21-}, by
\begin{align}{ 
\delta_G(\mu)=\left(\sum_{u,v}\,(u-v)^2 \mu(u,v)\right)^{1/2}.
\label{eq:CP}}
\end{align}

We say that a disturbance measure $\eta$ satisfies 
the {\em correspondence principle (CP)} if $\eta(B)=\delta_G(\mu)$ provided that  $B(\tau)$ 
and $B(0)$ have their JPD $\mu$ in the initial state $|{\psi,\xi}\rangle$.
An important property of the operator-based disturbance measure $\eta_{O}$ is that 
it satisfies the CP, as easily follows from Eq.~(\ref{eq:JPD-CP}). 
Similarly, the operator-based error measure $\varepsilon_{O}$ also satisfies CP 
in the sense that $\varepsilon_{O}(A)=\delta_G(\mu)$ provided that  $M(\tau)$ and 
$A(0)$ have their JPD $\mu$ in the initial state $|{\psi,\xi}\rangle$ as shown in 
Ref.~\cite{19A1}.

If $B(0)$ and $B(\tau)$ have their JPD $\mu$, the correlation between $B(0)$ and $B(\tau)$ has
the classical picture described by $\mu$, and the classical notion $\delta_G(\mu)$
of the root-mean-square deviation is applicable to quantifying the disturbance of $B$.  
In this case, according to the correspondence principle, any quantum definition of a disturbance
measure $\eta$ should be consistent with the classical measure $\delta_G$.  
Thus, we say that a quantum disturbance measure $\eta$ satisfies the correspondence principle
if two measures, the quantum $\eta$ and the classical $\delta_G$, are consistent, 
whenever the classical picture is available, as a desirable property of a quantum disturbance measure.
In this sense,  the correspondence principle determines the value of the disturbance on $B$, 
when the joint probability distribution of $B(0)$ and $B(\tau)$ exists,
in an analogous way as the probability distribution of $B(0)$ determines its standard deviation 
$\sigma(B)=\sigma(B(0))$ appearing in Eq.~(\ref {eq:UVEDR}).

\section{Disturbing observables without state change}\label{se:3DWOSC}

KJR \cite{KJR14}\  identified as `unphysical' the property of the operator-based disturbance 
measure $\eta_{O}$ that it does not assign the value 0 in  a case where 
the state has not changed at all.  
In such a case, the probability distribution of every observable has not changed,
so that this is a stronger violation of the DR.
However, we shall show here that this is not a peculiarity of the operator-based 
disturbance measure, but a straightforward consequence of the CP.

Consider a qubit measurement.
The projective measurement of $A=\sigma_z$
in the state $|{0}\rangle:=|{\sigma_z=+1}\rangle$ does not change the initial state $|{\psi}\rangle=|{0}\rangle$.
In this case, it was shown  \cite{05UUP} that the operator-based disturbance measure indicates that 
$B=\sigma_x$ is disturbed by the amount $\eta_{O}(\sigma_x)=\sqrt{2}$,
and this value was actually obtained by a neutron optical experiment \cite{13EVR}.
However,  according to the DR, every disturbance measure $\eta$ should assign the value 0,
and KJR \cite{KJR14} identified the above property of $\eta_{O}$ as a very unphysical property. 
In contrast, we shall show that every disturbance measure $\eta$ satisfying the CP
assigns the value $\sqrt{2}$.

\begin{figure}[h]
\begin{center}
\includegraphics[width=0.7\textwidth]{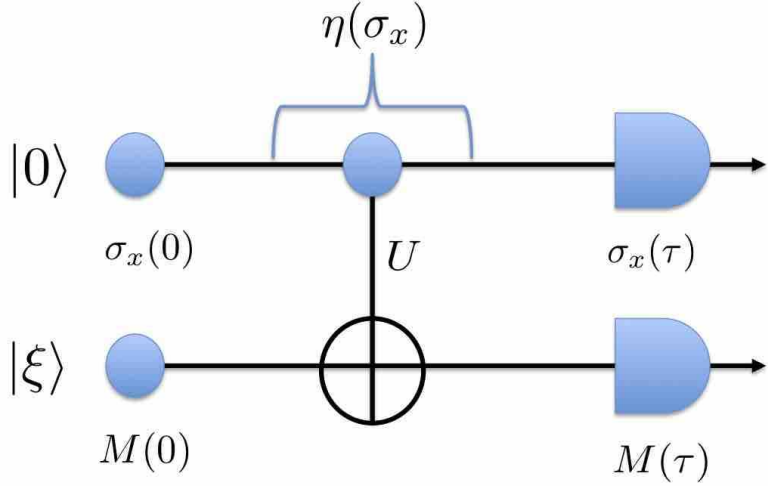}
\caption{
{\bf  Distributionally non-disturbing measurements are disturbing according to the correspondence principle.}
A projective measurement of $\sigma_z$ in $|{\psi}\rangle=|{0}\rangle$ is probability
non-disturbing to the observable $\sigma_x$. 
Thus, the DR  requires any disturbance measure to assign the value $0$.
However,  the CP\  requires any disturbance measure to assign the value $\sqrt{2}$.
}
\label{fig1}
\end{center}
\end{figure}

It is well-known that the projective measurement of $\sigma_z$ is carried out by
the controlled-NOT operation 
\begin{align}{U=|{0}\rangle\langle{0}|\otimes I +|{1}\rangle\langle{1}|\otimes \sigma_x}\end{align}
for the measured qubit $\mathbf{S}$ and the probe qubit $\mathbf{P}$ prepared
in the fixed state $|{\xi}\rangle=|{0}\rangle$ from $t=0$ to $t=\tau$
and by the subsequent meter measurement for $M=\sigma_z$ in $\mathbf{P}$ (Figure~\ref{fig1}) .
The Schr\"{o}dinger time evolution satisfies
\begin{align}
U(|{0}\rangle\otimes|{0}\rangle)&=|{0}\rangle\otimes|{0}\rangle,\\
U(|{1}\rangle\otimes|{0}\rangle)&=|{1}\rangle\otimes|{1}\rangle.
\end{align}
For $B=\sigma_x$ the Heisenberg time evolution is given by 
\begin{align}
\sigma_x(0)&=\sigma_x\otimes I,   \label{eq:Xzero}.\\
\sigma_x(\tau)&=\sigma_x\otimes\sigma_x \label{eq:Xtau}.
\end{align}
Here, Eq.~(\ref{eq:Xtau}) follows from
\begin{align*}
U^{\dagger}(\sigma_x\otimes I)U=|1\rangle\langle 1|\sigma_x|0\rangle\langle 0|\otimes\sigma_x
+|0\rangle\langle 0|\sigma_x|1\rangle\langle 1|\otimes\sigma_x
=\sigma_x\otimes\sigma_x .
\end{align*}
It follows that  $\sigma_x(\tau)$ and $\sigma_x(0)$ commute and they have the JPD $\mu(u,v)$ 
in the state  $|{\psi,\xi}\rangle=|{0,0}\rangle$ as
\begin{align}
\mu(u,v)&=\langle{0,0|P^{\sigma_x(\tau)}(u)P^{\sigma_x(0)}(v)|0,0}\rangle\nonumber\\
&=\langle{0,0|P^{\sigma_x\otimes\sigma_x}(u)P^{\sigma_x\otimes I}(v)|0,0}\rangle.
\end{align}
Then  we obtain
\begin{equation}\label{eq:DWOSC}
\mu(u,v)=\frac{1}{4}
\end{equation}
(cf.~Section~\ref{se:DWOSC}).
Thus, if the disturbance measure $\eta$ satisfies the CP, 
we have 
\begin{align}{
\eta(\sigma_x)^2=\delta_{G}(\mu)^2=\sum_{u,v=\pm 1}(u-v)^2\mu(u,v)
=2.
}\end{align}
Therefore we conclude $\eta(\sigma_x)=\sqrt{2}$.
Thus, the non-zero value $\eta_{O}(\sigma_x)=\sqrt{2}$ is not
a peculiar property of the operator-based disturbance measure.

It will be instructive to compare the above scenario (1) that the system is prepared in the sate 
$|0\rangle$ and then a projective measurement of $\sigma_z$ is performed and another
scenario (2) that the system is prepared in the sate $|0\rangle$ but no measurement is performed.
In both scenarios, the system state $|0\rangle$ is unchanged and the probability distribution
of any observable does not change.  How can we operationally distinguish the two scenarios.   
In scenario 1 we have shown that the observable $B=\sigma_x$ is disturbed.
For the time $t=0$ just before the measurement and the time $t=\tau$ just after the measurement, 
we obtain $B(0)=\sigma_x\otimes I$ and $B(\tau)=\sigma_x\otimes \sigma_x$ (cf.~Eqs.~(\ref{eq:Xzero}) 
and (\ref{eq:Xtau})).
Their joint probability distribution $p_2(u,v)$ satisfies $p_2(u,v)=1/4$ for any $u,v$ (cf.~Eq.~(\ref{eq:DWOSC}))
that leads to $\eta(B)=\sqrt{2}$.
On the other hand, scenario 2 is easily analyzed, so that we obtain  $B(0)=B(\tau)=\sigma_x\otimes I$,
and their joint probability distribution $p_1(u,v)$ satisfies $p_1(u,v)=\delta_{u.v}/2$
that leads to $\eta(B)=0$.
The joint probability distributions can be experimentally obtained by weak measurements and post-selections
as proposed by Lund and Wiseman \cite{LW10}.
Thus, we can operationally distinguish between the above two scenarios.

This conclusion might sound counter-intuitive, as the pure sate has 
the ``maximal information'' about the system.
However, unchanging the pure state does not imply unchanging the observable, 
because the ``maximal information'' about the system does not include 
the ``maximal information'' about an observable, 
analogously with the fact that the ``maximal information'' about 
the whole system does not include the ``maximal information'' 
about subsystems.

In fact, according to the available classical description, the conditional probability 
\begin{align}{
\Pr\{\sigma_x(\tau)=u|\sigma_x(0)=v\}=\mu(u|v)=\frac{1}{2}
}\end{align}
shows that the value of $\sigma_x(0)$ has been completely randomized, 
although their marginals have not changed at all as
\begin{align}{
\Pr\{\sigma_x(\tau)=u\}=\Pr\{\sigma_x(0)=u\}=\frac{1}{2}.
}\end{align}
Thus, the DR neglects the disturbance 
caused by the randomization by measurement without changing the probability distribution.

\section{State-dependent formulation for non-disturbing measurements}
We have shown that the DR with the notion of 
probability  non-disturbing measurements contradicts the CP.
To reconcile the conflict,
we shall characterize non-disturbing measurements 
from the two fundamental
requirements: the CP\ and the operational accessibility.

Consider the following condition.

(S) $B(\tau)$ and $B(0)$ have their JPD $\mu$ in $|{\psi,\xi}\rangle$ satisfying that
$\mu(u,v)=0$ if $u\ne v$.

From the point of view of the CP, if condition (S) holds,
we should conclude that the measurement $\mathbf{M}$ does not disturb $B$ in $|{\psi}\rangle$.
Thus, condition (S) is considered as a sufficient condition for a proper definition of non-disturbing
measurements.

On the other hand, from the point of view of operational accessibility, it is convenient to consider 
the {\em weak joint distribution (WJD)} $\nu(u,v)$ of $B(\tau)$ and $B(0)$ in $|{\psi,\xi}\rangle$ 
defined by
\begin{align}{
\nu(u,v)=\langle{\psi,\xi|P^{B(\tau)}(u)P^{B(0)}(v)|\psi,\xi}\rangle.
}\end{align}
The WJD always exists, though possibly takes negative or complex values,
and is operationally accessible by weak measurement and post-selection 
\cite{AAV88,Ste95,Joz07};  see also Ref.~\cite{11UUP} for a short survey. 
Then it is natural to consider the following condition. 

(W) The WJD of $B(\tau)$ and $B(0)$ in $|{\psi,\xi}\rangle$ satisfies that $\nu(u,v)=0$ 
if $u\ne v$.

If the measurement $\mathbf{M}$ does not disturb the observable $B$
in $|{\psi}\rangle$, any operational tests for witnessing the disturbance should fail.
Since measuring WJD is one of such operational tests for which the disturbance is detected 
if $\nu(u,v)\ne 0$ for some $u\ne v$  \cite{GWPP04,MLMSGW07}, 
condition (W) is considered as a necessary condition for a proper definition 
of non-disturbing measurements.

Obviously, (W) is logically weaker than or equivalent to (S). 
However, Theorem \ref{th:prop-non-distubing} (Section \ref{se:PT:S=W})
shows that both conditions are actually equivalent. 
In fact, according to the theory of quantum perfect correlations \cite{05PCN,06QPC}, 
both conditions (S) and (W) equivalently require that $B(\tau)$ and $B(0)$ are perfectly
correlated in the state $|{\psi,\xi}\rangle$ \cite{11UUP}.
Thus, the above argument justifies the following definition of non-disturbing 
measurements.
We say that the measurement $\mathbf{M}$ {\em is properly non-disturbing} to an observable $B$ in 
$|{\psi}\rangle$ if one of the conditions (S) or (W) is satisfied.  
Since the WJD is operationally accessible, this definition is also operationally accessible. 

\section{Reliability of the operator-based disturbance measure}
To consider the reliability of the operator-based disturbance measure,
we examine the following requirements:
(i) the CP, (ii) soundness, 
(iii) operational accessibility,
 and (iv) completeness.

We have already shown that the operator-based disturbance measure $\eta_O$ satisfies
the CP, i.e.,  $\eta_O(B)=\delta_{G}(\mu)$ if  $B(\tau)$ and $B(0)$ 
have the JPD $\mu$.
We introduce the {\em soundness} requirement:
Any disturbance measure $\eta$ should assign the value 0 to any properly non-disturbing measurements.
It is interesting to see that the CP\ implies soundness.
To show this, suppose that the measurement is properly non-disturbing to $B$ in $|{\psi}\rangle$.  
Then $B(\tau)$ and $B(0)$ have the JPD $\mu$ satisfying that $\mu(u,v)=0$ if $u\ne v$.
It follows that $\varepsilon_{G}(\mu)=0$ and by the CP\ we have
$\eta(B)=\varepsilon_{G}(\mu)=0$.  Accordingly, the operator-based disturbance measure
$\eta_O$ satisfies the soundness requirement.
We conclude, therefore, that even if the measurement does not change the state, the disturbance 
can be operationally detected as long as the operator-based disturbance measure takes a positive value.

It has been known that the operator-based disturbance measure $\eta_O$ is 
operationally accessible in the two ways: (i) the tomographic three state method, proposed by 
Ozawa~\cite{04URN} and experimentally realized by Erhalt et al.~\cite{12EDU} and others 
\cite{13VHE,13EVR,16A3} and (ii) the weak measurement method, proposed 
by Lund and Wiseman \cite{LW10} and experimentally realized by Rozema et al.~\cite{RDMHSS12} and others 
\cite{WHPWP13,RBBFBW14,14A1}.

As the converse of soundness,
a disturbance measure $\eta$ is said to be {\em complete}
if $\eta$ assigns the value 0 only to {\em properly} non-disturbing measurements.
There is an example in which $\eta_O$ does not satisfy completeness (Ref.~\cite{06QPC}, p.~750).  
However, it is known that $\eta_O$ satisfies completeness if (i) (commutative case) 
$B(\tau)$ and $B(0)$ commute
in $|{\psi,\xi}\rangle$ or if (ii) (dichotomic case) $B^2=I$ (Ref.~\cite{19A1}, Theorem 3).  

We have seen that the operator-based disturbance measure satisfies 
all requirements (i)--(iii), and partially satisfies requirement (iv) above.

Analogously from an argument for the operator-based error measure 
$\varepsilon_O$ in Ref.~\cite{19A1}, it follows that $\eta_O$ 
can be modified to satisfy completeness
by defining the {\em operator-based locally uniform disturbance measure} $\overline{\eta}$ 
as
\begin{align}{
\overline{\eta}(B,\mathbf{M},|{\psi}\rangle)=\sup_{t\in\mathbb{R}}\eta_O(B,\mathbf{M},e^{-itB}|{\psi}\rangle).
}\end{align}
Then the error measure $\overline{\eta}$ satisfies requirements (i) -- (iv)
and also (v) (Dominating property) 
$\eta_O(B,\mathbf{M},|{\psi}\rangle)\le \overline{\eta}(B,\mathbf{M},|{\psi}\rangle)$ for any $|{\psi}\rangle$,
and (vi) (Conservation property for dichotomic measurements)
$\overline{\eta}(B)=\eta_O(B)$ if $B^2=I$.
Thus, all the EDRs for $\eta_{O}$ also holds for $\overline{\eta}$; see analogous discussions for
the operator-based error measure in Ref.~\cite{19A1}. 

In the following we shall discuss another requirement on locality,
which the operator-based disturbance measure satisfies, but contradicts the DR. 

\section{Locality of disturbance}
\sloppy
We have argued that state-dependent formulations of error-disturbance relations are
well-founded by the operator-based disturbance measure, which is a sound disturbance
measure according to the notion of properly non-disturbing measurement that is supported by
the CP and operational accessibility, in contrast to KJR's claim that  the operator-based 
disturbance measure is not sound under the notion of distributionally non-disturbing measurements,
which we have shown to contradict the CP.

Yet, there is a prevailing view that only probability distributions of outcomes of 
measurements can be operationally compared  \cite{BLW14RMP},
despite the fact that the new experimental techniques enable us to operationally detect 
the change of an observable in time:
(i) the tomographic three state method \cite{04URN,12EDU,13VHE,13EVR,16A3} and 
(ii) the weak measurement method \cite{LW10,RDMHSS12,WHPWP13,RBBFBW14,14A1}.

In what follows, we shall show below another drawback of the DR that the notion of probability 
non-disturbing measurements violates a locality requirement to be posed below.

Consider a composite system $\mathbf{S}_1+\mathbf{S}_2$ in a state $|{\Psi}\rangle$.
Since any local measurement of $\mathbf{S}_2$ does not interact with the system $\mathbf{S}_1$,
we naturally take it for granted that any local measurement of $\mathbf{S}_2$ non-disturbing to an observable $B_2$ 
in $\mathbf{S}_2$ should be non-disturbing to the observable $B_1\otimes B_2$ for any observable $B_1$ in $\mathbf{S}_1$.
We call this requirement the {\em locality requirement} for a definition of disturbing measurements.
We shall show that the definition of distributionally non-disturbing measurements 
does not satisfy this requirement,  whereas the definition of properly non-disturbing measurements does satisfy the requirement as shown in Theorem \ref{th:locality} (Section \ref{se:Preserving}), 

For this purpose, we consider a  maximally entangled two-qubit system 
$\mathbf{S}_1+\mathbf{S}_2$ in the Bell state $|{\Phi^{+}}\rangle=(|{00}\rangle+|{11}\rangle)/\sqrt{2}$.  
Since $|{\Phi^{+}}\rangle=(|{0_x0_x}\rangle+|{1_x1_x}\rangle)/\sqrt{2}$,
the outcomes of the  joint local measurements of the observables 
$\sigma_x^{(1)}=\sigma_x\otimes I$ and $\sigma_x^{(2)}=I\otimes \sigma_x$ show a perfect correlation.
From Theorem \ref{th:jpd-preserving} (Section \ref{se:Preserving}),
measurements properly non-disturbing to $\sigma_x^{(2)}$ does not change the JPD
of $\sigma_x^{(1)}$ and $\sigma_x^{(2)}$, so that the perfect correlation 
between $\sigma_x^{(1)}$ and $\sigma_x^{(2)}$ is not disturbed.
However, we shall show that a probability  non-disturbing measurement breaks
the perfect correlation, and this concludes that the definition of probability  
non-disturbing measurements does not satisfy the locality requirement, 
according to Theorem \ref{th:preserving} (Section \ref{se:Preserving}).

\begin{figure}[h]
\begin{center}
\includegraphics[width=0.8\textwidth]{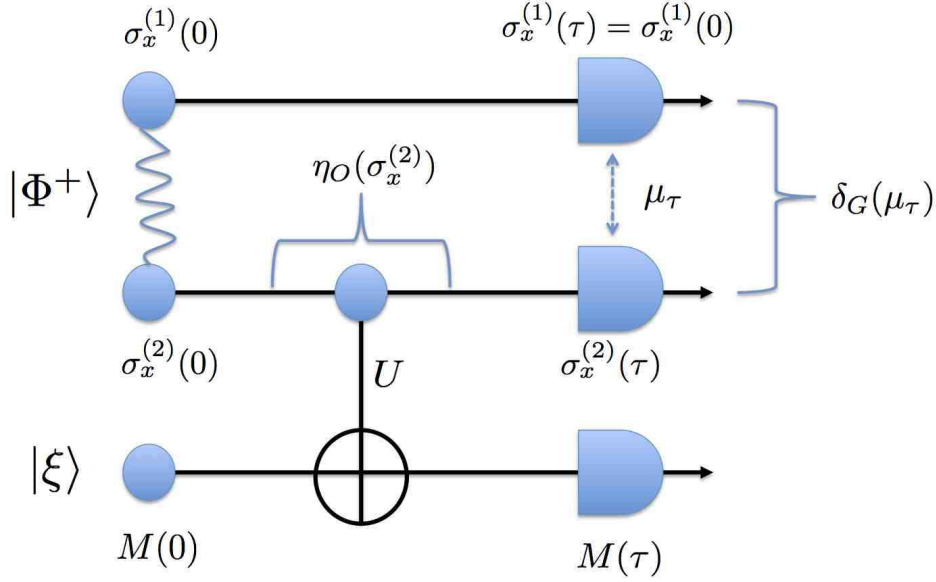}
\caption{{\bf  Definition of distributionally non-disturbing measurements violates
the locality requirement.}
(i)  The projective $\sigma_z^{(2)}$ measurement in $|{\Phi^{+}}\rangle$  is 
probability  non-disturbing to $\sigma_x^{(2)}$,  
but disturbs the JPD $\mu$ between $\sigma_x^{(1)}$ and $\sigma_x^{(2)}$ in $|{\Phi^{+}}\rangle$.
The perfect correlation between $\sigma_x^{(1)}$ and $\sigma_x^{(2)}$ at time 0, i.e.,  
$\delta_{G}(\mu_0)=0$,  is disturbed by the amount $\delta_{G}(\mu_{\tau})=\eta_O(\sigma_x^{(2)})=\sqrt{2}$.
(ii) The projective $\sigma_{\theta}^{(2)}$ measurement  in $|{\Phi^{+}}\rangle$ is
distributionally non-disturbing to $\sigma_x^{(2)}$, 
but disturbs the JPD $\mu$ between $\sigma_x^{(1)}$ and $\sigma_x^{(2)}$  in $|{\Phi^{+}}\rangle$
for $0\le\theta<\pi/2$.
The perfect correlation between $\sigma_x^{(1)}$ and $\sigma_x^{(2)}$, i.e., $\delta_{G}(\mu_0)=0$,  
is disturbed by the amount $\delta_{G}(\mu_{\tau})=\eta_O(\sigma_x^{(2)})=\sqrt{2}\cos\theta$.
(iii) An arbitrary local measurement of $\mathbf{S}_2$ in $|{\Phi^{+}}\rangle$ 
with the disturbance $\eta_O(\sigma_x^{(2)})$ disturbs the perfect correlation 
between $\sigma_x^{(1)}$ and $\sigma_x^{(2)}$,  i.e., $\delta_{G}(\mu_0)=0$,  
by the amount $\delta_G(\mu_{\tau})=\eta_O(\sigma_x^{(2)})$. 
This relation leads to a security tradeoff relation 
for the E91 quantum cryptography protocol \cite{Eke91}.
}
\label{fig}
\end{center}
\end{figure}

\subsection*{(i) Projective $\sigma_z^{(2)}$ measurement.} 

Suppose that the observer makes a 
projective $\sigma_z^{(2)}$ measurement 
just before the joint local  measurements of $\sigma_x^{(1)}$ and $\sigma_x^{(2)}$
(Figure \ref{fig} (i)).  
The measuring interaction is given by
\begin{align}{U=I\otimes |{0}\rangle\langle{0}|\otimes I +I\otimes|{1}\rangle\langle{1}|\otimes \sigma_x,} \end{align}
turned on from $t=0$ to $t=\tau$ between $\mathbf{S}_1+\mathbf{S}_2$ and 
the probe $\mathbf{P}=\mathbf{S}_3$ prepared in $|{\xi}\rangle=|{0}\rangle$
with the meter $M=\sigma_z^{(3)}$.
The time evolutions of relevant observables are given by
\begin{align}{
\sigma_x^{(1)}(0)&=\sigma_x\otimes I\otimes I,\\
\sigma_x^{(1)}(\tau)&=\sigma_x\otimes I\otimes I,\\
\sigma_x^{(2)}(0)&=I\otimes\sigma_x\otimes I,\\
\sigma_x^{(2)}(\tau)&=I\otimes \sigma_x\otimes\sigma_x.
}\end{align}

Then we shall see that the projective $\sigma_z^{(2)}$ measurement 
is distributionally non-disturbing to $\sigma_x^{(2)}$, but
it disturbs the perfect correlation 
between $\sigma_x^{(1)}$ and $\sigma_x^{(2)}$.
To show that, let $\mu_t$ be the JPD of $\sigma_x^{(1)}(t)$ and $\sigma_x^{(2)}(t)$ for $t=0,\tau$.
Then  we have 
\begin{align}{\label{eq:ent-1-1}
\mu_0(u,v)=\frac{1}{2}\delta_{u,v}, \quad  \mu_{\tau}(u,v)=\frac{1}{4}}\end{align}
for any $u,v=\pm 1$ (cf.~Section \ref{se:ent-1}).  
Since the marginal probability for $\sigma_x^{(2)}$ does not change, 
the projective $\sigma_z^{(2)}$ measurement is distributionally non-disturbing to $\sigma_x^{(2)}$. 
However, the perfect correlation between $\sigma_x^{(1)}$ and $\sigma_x^{(2)}$ at time $t=0$
has been disturbed.
The amount of the disturbance of the perfect correlation, i.e., $\delta_{G}(\mu_{0})=0$,  is measured by
the classical root-mean-square deviation $\delta_{G}(\mu_{\tau})$, and
we have 
\begin{align}{\label{eq:ent-1-2}
\delta_{G}(\mu_{\tau})=\eta_O(\sigma_x^{(2)})=\sqrt{2}}
\end{align}
(cf.~Section \ref{se:ent-1}).
Thus, the projective $\sigma_z^{(2)}$ measurement is distributionally disturbing to
$B_1\otimes \sigma_x^{(2)}$ for some observable $B_1$ of $\mathbf{S}_1$ by Theorem  
\ref{th:preserving} (Section \ref{se:Preserving}).
Therefore, we conclude that the definition of distributionally non-disturbing measurements
does not satisfy the locality requirement.

Since all observables $\sigma_x^{(1)}(0),\sigma_x^{(1)}(\tau),\sigma_x^{(2)}(0),\sigma_x^{(2)}(\tau)$ 
are mutually commuting,
we have their joint probability distribution.
The relation $\sigma_x^{(2)}(0)=\sigma_x^{(1)}(0)$ holds with probability one
by entanglement, and $\sigma_x^{(1)}(\tau)=\sigma_x^{(1)}(0)$ holds by locality
of the measurement. Thus, we have
\begin{align}{\Pr\{\sigma_x^{(1)}(\tau)=\sigma_x^{(2)}(0)=\sigma_x^{(1)}(0)\}=1.} \end{align}
From this and the relation $\mu_{\tau}(u,v)=1/4$ above,
 we obtain
\begin{align}{
\Pr\{\sigma_x^{(2)}(\tau)=v',\sigma_x^{(1)}(\tau)=u',\sigma_x^{(2)}(0)=v,\sigma_x^{(1)}(0)=u\}
=\frac{1}{4}\delta_{u,v}\delta_{u,u'}.
}\end{align}
Thus, 
we have the conditional probability 
\begin{align}{\Pr\{\sigma_x^{(2)}(\tau)=v'|\sigma_x^{(2)}(0)=v\}=\frac{1}{2}}\end{align}
showing that $\sigma_x^{(2)}$ is completely randomized by the measuring interaction,
whereas the DR neglects this randomization manifest in the joint probability
$\mu_{\tau}(u,v)$ of outcomes of local measurements of $\sigma_x^{(1)}(\tau)$ and 
$\sigma_x^{(2)}(\tau)$.

\subsection*{(ii) Projective $\sigma_{\theta}^{(2)}$ measurement.}
For quantitative considerations,  suppose that the observer makes a 
projective $\sigma_{\theta}^{(2)}$ measurement just before the joint 
local measurements of $\sigma_x^{(1)}$ and $\sigma_x^{(2)}$,
where $\sigma_{\theta}=\cos\theta\sigma_z+\sin\theta\sigma_x$ for $0\le \theta< \pi/2$
(Figure \ref{fig} (ii)).
Then the projective $\sigma_{\theta}$ measurement is distributionally non-disturbing to 
the observable $\sigma_x^{(2)}$
(cf.~Section \ref{se:ent-3}). 
However, they disturb the perfect correlation between $\sigma_x^{(1)}$
and $\sigma_x^{(2)}$.
In fact, the JPD $\mu_{\tau}$ of 
$\sigma_x^{(1)}(\tau)$ and $\sigma_x^{(2)}(\tau)$ is given by 
\begin{align}{\label{eq:entanglement-3}
\mu_{\tau}(u,v)
&=\frac{1}{4}\delta_{u,v}(1+\sin^2\theta)+\frac{1}{4}(1-\delta_{u,v})\cos^2\theta
}\end{align}
and the classical root-mean-square deviation $\delta_G(\mu_{\tau})$ and the disturbance $\eta_O(\sigma_x^{(2)})$ 
are given by
\begin{align}{\label{eq:entanglement-4}
\delta_G(\mu_\tau)=\eta_{O}(\sigma_x^{(2)})=\sqrt{2}\cos \theta
}\end{align}
 (cf.~Section \ref{se:ent-3}).  
 Thus, the joint probability distribution of the outcomes of
joint local measurements of $\sigma_x^{(1)}(\tau)$ and $\sigma_x^{(2)}(\tau)$ favors the non-zero value  
$\eta_{O}(\sigma_x^{(2)})=\sqrt{2}\cos \theta$, in contrast to the DR requiring
$\eta_{O}(\sigma_x^{(2)})=0$.
  
 \subsection*{(iii) Arbitrary local measurements.}
Suppose that the observer makes an arbitrary local measurement $\mathbf{M}$ of $\mathbf{S}_2$
from $t=0$ to $t=\tau$ with the probe prepared in $|{\xi}\rangle$
just before the joint local measurements of $\sigma_x^{(1)}$ and  $\sigma_x^{(2)}$ (Figure \ref{fig} (iii)).  
Then the JPD $\mu_0$ and 
the classical root-mean-square deviation $\delta_{G}(\mu_0)$ satisfy $\mu_0(u,v)=\delta_{u,v}/2$ 
and $\delta_G(\mu_0)=0$.
From Theorem \ref{th:chain-rule} (ii) (Section \ref{se:G=O}), the relation
\begin{align}{\label{eq:G=O}
\delta_G(\mu_{\tau})=\eta_{O}(\sigma_x^{(2)})
}\end{align}
holds for any local measurement $\mathbf{M}$ of $\mathbf{S}_2$.
Since $\eta_{O}(\sigma_x^{(2)})=0$ if and only if $\mathbf{M}$ is properly non-disturbing to $\sigma_x^{(2)}$
from Theorem \ref{th:chain-rule} (iii) (Section \ref{se:G=O}),
we conclude $\delta_G(\mu_{\tau})=0$ 
if and only if $\mathbf{M}$ is properly non-disturbing to $\sigma_x^{(2)}$.

Since Eq.~(\ref{eq:G=O}) holds for an arbitrary local measurement, it  
has an interesting application to quantum cryptography protocol E91 \cite{Eke91}.
Suppose that Alice and Bob share a maximally entangled pair $\mathbf{S}_1+\mathbf{S}_2$ in $|{\Phi^{+}}\rangle$ 
and that Eve measures $\mathbf{S}_2$ for eavesdropping the shared key.
Suppose that Alice and Bob share a key encoded in $\sigma_z^{(1)}$ and $\sigma_z^{(2)}$.
To estimate how much information leaks to Eve,  cooperative Alice and Bob measure
the error probability $P_{e}^{AB}$ defined by
 $P_{e}^{AB}=\delta_G(\mu)^2/4$.
Let $\varepsilon_O(\sigma_z^{(2)})$ be Eve's error for $\sigma_z^{(2)}$ measurement and let $\eta_O(\sigma_x^{(2)})$
be Eve's disturbance caused on $\sigma_x^{(2)}$.  Then  Eq.~(\ref{eq:G=O}) serves as a bridge between
$P_{e}^{AB}$ and the disturbance $\eta_O(\sigma_x^{(2)})$,
and the error-disturbance relation further relates $P_{e}^{AB}$ with Eve's 
error probability $P_{e}^{E}$ for eavesdropping on the key defined by
$P_{e}^{E}=\varepsilon_O(\sigma_z^{(2)})^2/4$, as follows.
Recall that the tight EDR
\begin{align}{
(\varepsilon_O(\sigma_z^{(2)})^2-2)^2+(\eta_O(\sigma_x^{(2)})^2-2)^2\le 4
}\end{align}
holds for $\varepsilon_O(\sigma_z^{(2)})$ and $\eta_O(\sigma_x^{(2)})$  (Ref.~\cite{14EDR}, Eq.~(28)).
Then this optimizes Eve's error probability $P_{e}^{E}$  as
\begin{align}{
P_{e} ^{E}({\rm optimal})= \frac{1}{2}\!-\!\sqrt{\frac{1}{4}\!-\!\left(P_{e}^{AB}\!-\!\frac{1}{2}\right)^2}.
}\end{align}
Thus, if the entanglement is not disturbed, i.e., $P_{e}^{AB}=0$, then Alice and Bob conclude
$P_{e}^{E}({\rm optimal})=1/2$ to ensure that no information has leaked to Eve.
On the other hand, if Eve makes the projective measurement of $\sigma_z^{(2)}$ with $\varepsilon_O(\sigma_z^{(2)})=0$, 
then she has the complete information $P_{e} ^{E}=P_{e}^{E}({\rm optimal})=0$ but this is detected by
Alice and Bob as $P_{e}^{AB}=1/2$ and $\eta_O(\sigma_x^{(2)})=\sqrt{2}$.  
However, the DR forces any disturbance measure $\eta$ to assign $\eta(\sigma_x^{(2)})=0$.
How does the DR work to analyze the security of quantum communication?

\section{Defense of state-dependent formulations}

In order to examine the reliability of the
operator-based disturbance measure,
KJR \cite{KJR14}\  introduced the following definition.
A state $|{\psi}\rangle$ is called a {\em zero-noise zero-disturbance (ZNZD) state} 
with respect to observables $A$ and $B$ 
if the projective measurement of $A$ in the state $|{\psi}\rangle$,
which always satisfies $\varepsilon(A)=0$,
is distributionally non-disturbing to $B$.
Then they proved that for every pair of non-commuting observables
$A$ and $B$, there exists a  ZNZD state
$|{\psi}\rangle$ such that $|\langle{\psi|[A,B]|\psi}\rangle|\not=0$.
Thus, if the disturbance measure $\eta$ satisfies the DR, 
any relation of the form
\begin{align}
\sum_{m,n=0}^{\infty}f_{mn}(A,B)\varepsilon(A,\rho)^m\eta(B,\eta)^n\ge
|{\langle {[A,B]} \rangle}|,
\end{align}
where $f_{00}(A,B)=0$, must be violated.
From this, KJR \cite{KJR14}\   concluded that any state-dependent EDR, based on
the expectation value of the commutator as a lower bound, is not tenable, and that
state-independent formulations are inevitable.

We have two objections to their claims.  First of all, the universally valid relation 
(\ref{eq:UVEDR}) with $\varepsilon_{O}(A)=0$
leads to the relation 
\begin{align}
\eta_{O}(B)\ge\frac{|{\langle {[A,B]} \rangle}|}{2\sigma(A)}>0
\end{align}
for any projective measurement of $A$ in any ZNZD state such that $|\langle{\psi|[A,B]|\psi}\rangle|\not=0$.
Thus, the measurement is properly disturbing to $B$ by the soundness of $\eta_O$,
and consequently the disturbance is operationally detectable by the operational accessibility
of the definition of properly non-disturbing measurements, 
so that the assumption by KJR \cite{KJR14}\  that $\eta(B)=0$ in any ZNZD state is unfounded. 

Secondly, they concluded that state-independent formulations are inevitable for alternative formulations.
However, currently proposed state-independent formulations of EDRs 
\cite{App98c,BLW14RMP,App16} do not appear to capture the essence of  
Heisenberg's original idea.  
Recall that Heisenberg derived his EDR by the $\gamma$-ray microscope thought experiment,
in which the EDR is derived from the relation between the resolution power and the Compton recoil,
reciprocally relating to the wave length of the incident light.
Since the wave length is independent of the state of the object,
the above formulation might be considered as state-independent.
However, the analysis is valid only state-dependently, since the resolution power of the 
microscope can be defined by the wave length only in the limited situation in which 
the object is properly placed in the scope of the microscope.  Thus, we can adequately 
define the error of the $\gamma$-ray microscope only state-dependently.  
In the state-independent formulations,
currently one defines the state-independent error as the worst case of the state-dependent error, 
which must diverge to infinity as the object wave function spreads out of, or moves far apart
from, the scope of the microscope.  
Such state-independent definitions would facilitate to reproduce the form of Heisenberg's 
original formulation,  but do not keep the physics underlying it.
Thus, state-dependent formulations are inevitable to represent Heisenberg's original 
idea underlying the uncertainty principle. 

\section{Discussion}
In this paper, we have given a definition of non-disturbing measurement from the point of view
of the correspondence principle and operational accessibility.
Subsequently, we have established the reliability of the operator-based disturbance measure.
We have already discussed the reliability of the operator-based error measure in our
previous work \cite{19A1}.
Both accounts ensure that universally valid EDRs \cite{03UVR,Bra13,Bra14,14EDR} 
reliably represent a dynamical aspect of Heisenberg's uncertainty principle besides
the well-established relation for the indeterminacy in quantum states representing 
a kinetic aspect of the principle.
Thus, the objections to state-dependent formulations of EDRs shown in \cite{BLW14RMP,KJR14}
are unfounded, although those views appear to still prevail in the literature 
\cite{Ren17,Mao19}.
We conclude that the theory \cite{03UVR,04URN,LW10,Bra13,Bra14,14EDR,19A1}
and experiments \cite{12EDU,RDMHSS12,13EVR,13VHE,14A1}
for state-dependent formulations of EDRs
are reliable and that state-dependent formulations
are inevitable to represent Heisenberg's original idea underlying the uncertainty principle. 

The new quantitative methods developed in this paper for universally valid EDRs
 with the well-defined operator-based disturbance measure incorporating with the
methods of weak values and weak measurements will provide new quantitative methods to
understand the change, transfer, or disturbance of observables in time, 
which does not manifest in the change of the probability distribution, but which does
manifest in the time-like correlation.
This quantity will be useful and even inevitable for exploring foundational problems 
in quantum physics  
including the long-lasting controversy over the roles of uncertainty principle 
in which-way measurements for interferometers 
(Refs.~\cite{SEW91,STCW94,GWPP04,MLMSGW07,Xia19} and the references therein).
In addition to the foundational problems, it will be expected that universally valid EDRs 
call for new research interests in exploring various frontiers in physics 
including fault-tolerant quantum computing \cite{NC00,02CQC,TT20}, 
quantum metrology \cite{88MS,GLM04,Asp14}, and multi-messenger astronomy \cite{Bar17}, 
in which technological limits would be overcome by the fundamental principle 
independent of particular models.
We hope that the methods of operator-based disturbance measures as well as 
operator-based error measures will be accepted for broad areas of quantum physics. 

\section{Projective $\sigma_{z}$ measurement}\label{se:DWOSC}
Here, we shall give a derivation of Eq.~(\ref{eq:DWOSC}).
The JPD $\mu$ of $\sigma_x(\tau)=\sigma_x\otimes\sigma_x$ and $\sigma_x(0)=\sigma_x\otimes I$ in the state 
$|{0,0}\rangle$ is given by
\begin{align*}
\mu(u,v)
&=\langle{0,0|P^{\sigma_x\otimes \sigma_x}(u)P^{\sigma_x\otimes I}(v)|0,0}\rangle.
\end{align*}
We have
\begin{align*}
P^{\sigma_x\otimes\sigma_x}(+1)P^{\sigma_x\otimes I}(\pm 1)
&=
[P^{\sigma_x}(+1)\otimes P^{\sigma_x}(+1)+ P^{\sigma_x}(-1)\otimes P^{\sigma_x}(-1)]
(P^{\sigma_x}(\pm 1)\otimes I) \\
&=
P^{\sigma_x}(\pm1)\otimes P^{\sigma_x}(\pm1),\\
P^{\sigma_x\otimes\sigma_x}(-1)P^{\sigma_x\otimes I}(\pm 1)
&=
[P^{\sigma_x}(+1)\otimes P^{\sigma_x}(-1)+P^{\sigma_x}(-1)\otimes P^{\sigma_x}(+1)]
(P^{\sigma_x}(\pm 1)\otimes I)\\
&=
P^{\sigma_x}(\pm1)\otimes P^{\sigma_x}(\mp1).
\end{align*}
Consequently,
\begin{align*}
\mu(+1,\pm 1)
&=\langle{0,0|P^{\sigma_x}(\pm 1)\otimes P^{\sigma_x}(\pm1)|0,0}\rangle
=\langle{0|P^{\sigma_x}(\pm 1)|0}\langle{0|P^{\sigma_x}(\pm1)|0}\rangle\\
&=\frac{1}{4},\\
\mu(-1,\pm 1)
&=\langle{0,0|P^{\sigma_x}(\pm1)\otimes P^{\sigma_x}(\mp 1)|0,0}\rangle
=\langle{0|P^{\sigma_x}(\pm 1)|0}\langle{0|P^{\sigma_x}(\mp 1)|0}\rangle\\
&=\frac{1}{4}.
\end{align*}
Therefore, we obtain Eq.~(\ref{eq:DWOSC}), i.e., 
\[
\mu(u,v)=\frac{1}{4}.
\]

\section{Equivalence for properly non-disturbing measurements}
\label{se:PT:S=W}

\begin{Theorem}\label{th:prop-non-distubing}
Let $\mathbf{M}$ be a measurement of a system $\mathbf{S}$ in a state $|{\psi}\rangle$
carried out by a measuring interaction with a probe $\mathbf{P}$ 
prepared in a fixed state $|{\xi}\rangle$ from $t=0$ to $t=\tau$.  Then for any observable $B$ in $\mathbf{S}$,
the following conditions are equivalent.

(i) Condition (W): The WJD $\nu$ of $B(\tau)$ and $B(0)$ in $|{\psi,\xi}\rangle$ satisfies that $\nu(u,v)=0$ 
if $u\ne v$.

(ii) The relation
\begin{align*}{
P^{B(\tau)}(u)|{\psi,\xi}\rangle=P^{B(0)}(u)|{\psi,\xi}\rangle
}\end{align*}
holds for any $u$.

(iii) Condition (S): $B(\tau)$ and $B(0)$ have their JPD $\mu$ in $|{\psi,\xi}\rangle$ satisfying that
$\mu(u,v)=0$ if $u\ne v$.
\end{Theorem}
\begin{proof} 
The assertion was generally proved in Refs.~\cite{05PCN,06QPC}
after a lengthy argument.  
We give a direct proof for the present context.

(i)$\Rightarrow$(ii):
Suppose (i) holds.
Then the WJD $\nu(u,v)$ of $B(\tau)$ and $B(0)$ in $|{\psi,\xi}\rangle$
satisfies $\nu(u,v)=0$ if $u\ne v$.
It follows that  $\nu(u,u)=\sum_{v}\nu(u,v)$.  Thus,
\begin{eqnarray*}
\langle{\psi,\xi|P^{B(\tau)}(u)P^{B(0)}(u)|\psi,\xi}\rangle
&=&\langle{\psi,\xi|P^{B(0)}(u)|\psi,\xi}\rangle,\\
\langle{\psi,\xi|P^{B(0)}(u)P^{B(\tau)}(u)|\psi,\xi}\rangle
&=&\langle{\psi,\xi|P^{B(\tau)}(u)|\psi,\xi}\rangle.
\end{eqnarray*}
Consequently,
\begin{align*}{
\|P^{B(\tau)}(u)|{\psi,\xi}\rangle-P^{B(0)}(u)|{\psi,\xi}\rangle\|^2=0,
}\end{align*}
and 
\begin{eqnarray*}
P^{B(\tau)}(u)|{\psi,\xi}\rangle=P^{B(0)}(u)|{\psi,\xi}\rangle.
\end{eqnarray*}
Thus, condition (ii) holds and the implication (i)$\Rightarrow$(ii) follows.

(ii)$\Rightarrow$(iii):
Suppose (ii) holds.
Then
\begin{eqnarray*}
P^{B(0)}(u)P^{B(\tau)}(v)|{\psi,\xi}\rangle&=&\delta_{u,v}P^{B(0)}(u)|{\psi,\xi}\rangle,\\
P^{B(\tau)}(v)P^{B(0)}(u)|{\psi,\xi}\rangle&=&\delta_{u,v}P^{B(0)}(u)|{\psi,\xi}\rangle.
\end{eqnarray*}
Consequently,
\[
P^{B(0)}(u)P^{B(\tau)}(v)|{\psi,\xi}\rangle=P^{B(\tau)}(v)P^{B(0)}(u)|{\psi,\xi}\rangle.
\]
It follows that $B(0)$ and $B(\tau)$ commute in $|{\psi,\xi}\rangle$ and condition 
(S) holds.
Thus the implication (ii)$\Rightarrow$(iii) follows.

Since the implication (iii)$\Rightarrow$(i) holds obviously, all conditions (i) -- (iii) are equivalent.
\end{proof}

\section{Locality of properly non-disturbing measurements}
\label{se:Preserving}

\begin{Theorem}\label{th:locality}
The definition of properly disturbing measurements satisfies the locality requirement.
\end{Theorem}
\begin{proof}
Let $\mathbf{M}$ be a local measurement of $\mathbf{S}_2$ in a composite system $\mathbf{S}_1+\mathbf{S}_2$ 
in a state $|{\Psi}\rangle$.
Without any loss of generality that $\mathbf{M}$ is carried out by a measuring interaction 
$U$ with a probe $\mathbf{P}$ prepared in a state $|{\xi}\rangle$ from time $t=0$ to $t=\tau$.  
Suppose that $\mathbf{M}$ is properly non-disturbing to an observable $B_2$ in $\mathbf{S}_2$.
Let $g(v)$ be a polynomial in $v$. 
From Theorem \ref{th:prop-non-distubing} (ii) (Section \ref{se:PT:S=W}) we have
\begin{align*}{
g(B_2(\tau))|{\Psi,\xi}\rangle=g(B(0))|{\Psi,\xi}\rangle.
}\end{align*}
Let $B_1$ be an observable in $\mathbf{S}_1$.  Let $f(u)$ be a polynomial in $u$. 
Since $B_1(0)=B_1(\tau)$ by the locality of $\mathbf{M}$ we have 
\begin{align*}{
f(B_1(\tau))g(B_2(\tau))|{\Psi,\xi}\rangle=f(B_1(0))g(B_2(0))|{\Psi,\xi}\rangle.
}\end{align*}
It follows from linearity that
\begin{align*}{
h(B_1(\tau),B_2(\tau))|{\Psi,\xi}\rangle=h(B_1(0),B_2(0))|{\Psi,\xi}\rangle
}\end{align*}
for any polynomial $h(u,v)$ in $(u,v)$, and in particular we have
\begin{align*}{
P^{B_1(\tau)B_2(\tau)}(v)|{\Psi,\xi}\rangle=P^{B_1(0)B_2(0)}(v)|{\Psi,\xi}\rangle.
}\end{align*}
Thus, $\mathbf{M}$ does not disturb $B_1\otimes B_2$ for any $B_1$ in $\mathbf{S}_1$
by Theorem \ref{th:prop-non-distubing} (ii) (Section \ref{se:PT:S=W}).
Therefore, the definition of properly disturbing measurements satisfies the locality requirement.
\end{proof}

\begin{Theorem}\label{th:preserving}
Let $\mathbf{S}_1+\mathbf{S}_2$ be a composite system in a state $|{\Psi}\rangle$.
Let $B_1$ be an observable in $\mathbf{S}_2$.
Any local measurement $\mathbf{M}$ of $\mathbf{S}_2$ distributionally non-disturbing to $B_1\otimes B_2$ 
for any $B_1$ in $\mathbf{S}_1$ does not change the JPD of observable $B_1$ and $B_2$ for
any observable $B_1$ in $\mathbf{S}_1$.
\end{Theorem}
\begin{proof}
Let $B_1$ be an observable in $\mathbf{S}_1$ and 
$A_1=P^{B_1}(u)$.
By assumption, $\mathbf{M}$ does not change the probability distribution of 
$A_1\otimes B_2=P^{B_1}(u)\otimes B_2$, so that all moments of  
$A_1\otimes B_2$ are unchanged as 
\begin{align*}{
\langle{\Psi,\xi|P^{B_1(\tau)}(u) B_2(\tau)^n|\Psi,\xi}\rangle
=\langle{\Psi,\xi|P^{B_1(0)}(u) B_2(0)^n|\Psi,\xi}\rangle
} \end{align*}
for all $n$.
By linearity, we have
\begin{align*}{
\langle{\Psi,\xi|P^{B_1(0)}(u)f(B_2(0))|\Psi,\xi}\rangle
=\langle{\Psi,\xi|P^{B_1(\tau)}(u)f(B_2(\tau))|\Psi,\xi}\rangle
} \end{align*}
for any polynomial $f(w)$ in $w$.  It follows that 
\begin{align*}{
\langle{\Psi,\xi|P^{B_1(\tau)}(u)P^{B_2(\tau)}(v)|\Psi,\xi}\rangle
=\langle{\Psi,\xi|P^{B_1(0)}(u)P^{B_2(0)}(v)|\Psi,\xi}\rangle,
}\end{align*}
and hence $\mathbf{M}$ does not change the JPD of $B_1$ and $B_2$ for
any observable $B_1$ in $\mathbf{S}_1$.
\end{proof}

\begin{Theorem}\label{th:jpd-preserving}
Let $\mathbf{S}_1+\mathbf{S}_2$ be a composite system in a state $|{\Psi}\rangle$.
Any local measurement $\mathbf{M}$ of $\mathbf{S}_2$ properly non-disturbing to $B_2$ in $\mathbf{S}_2$ 
does not change the JPD of observable $B_1$ and $B_2$ for
any observable $B_1$ in $\mathbf{S}_1$.
\end{Theorem}
\begin{proof}
Any local measurement $\mathbf{M}$ of $\mathbf{S}_2$ properly non-disturbing to $B_2$ in $\mathbf{S}_2$
is properly non-disturbing to $B_1\otimes B_2$ for any $B_1$ in $\mathbf{S}_1$
by Theorem \ref{th:locality} (Section \ref{se:Preserving}),
and hence it is distributionally non-disturbing to $B_1\otimes B_2$ for any $B_1$ in $\mathbf{S}_1$.
Consequently, the assertion follows from Theorem \ref{th:preserving}. 
\end{proof}

\section{Operator-based disturbance measure and disturbance of entanglement}
\label{se:G=O}
\begin{Theorem}\label{th:chain-rule}
Let $\mathbf{S}_1+\mathbf{S}_2$ be a composite system in a state $|{\Psi}\rangle$.
Let $\mathbf{M}$ be a local measurement of the system $\mathbf{S}_2$ 
carried out by a measuring interaction with a probe $\mathbf{P}$ 
prepared in a fixed state $|{\xi}\rangle$ from $t=0$ to $t=\tau$. 
Let $\mu_t(u,v)$ be the JPD of an observable $B_1(t)$ in $\mathbf{S}_1$ and
an observable $B_2(t)$ in $\mathbf{S}_2$ for $t=0,\tau$.
Let $\eta_O(B_2)$ be the operator-based disturbance of $\mathbf{M}$ for $B_2$.
Let $\delta_G(\mu_t)$ be the classical root-mean-square deviation determined by $\mu_t$.
Then we have the following.

(i) The relation
\begin{align*}{
|\delta_G(\mu_{\tau})-\delta_G(\mu_0)|&\le \eta_O(B_2)\le  \delta_G(\mu_{\tau})+\delta_G(\mu_0).
}\end{align*}
holds.

(ii) If $\delta_G(\mu_0)=0$ then $\delta_G(\mu_{\tau})=\eta_O(B_2)$.

(iii)  If $\delta_G(\mu_0)=0$ and $B_2^2=I$ then $\mathbf{M}$ is properly non-disturbing to
$B_2$ if and only
 $\delta_G(\mu_{\tau})=0$
\end{Theorem}
\begin{proof}
(i) We have the relations
\begin{align*}{
\delta_G(\mu_t)&=\|B_1(t)|{\Psi,\xi}\rangle-B_2(t)|{\Psi,\xi}\rangle\|\\
\eta_O(B_2)&=\|B_2(\tau)|{\Psi,\xi}\rangle-B_2(0)|{\Psi,\xi}\rangle\|,
}\end{align*}
and hence assertion (i) follows from repeated uses of the triangular inequality.

(ii) Follows by substituting $\delta_G(\mu_0)=0$ in (i).

(iii) Follows from (ii) and the completeness of $\eta_O$ for dichotomic observables.
\end{proof}

\bigskip

\subsection{Projective $\sigma_{z}^{(2)}$ measurement }\label{se:ent-1}
Consider the projective measurement of $A=\sigma_z^{(2)}$ in $\mathbf{S}_1+\mathbf{S}_2$ 
carried out by the measuring interaction
\begin{align*}{U=I\otimes |{0}\rangle\langle{0}|\otimes I +I\otimes|{1}\rangle\langle{1}|\otimes \sigma_x} \end{align*}
turned on from $t=0$ to $t=\tau$ between $\mathbf{S}_1+\mathbf{S}_2$ and 
the probe $\mathbf{P}=\mathbf{S}_3$ prepared in $|{\xi}\rangle=|{0}\rangle$
with the meter observable $M=\sigma_z^{(3)}$.
Consider the Heisenberg operators $\sigma_x^{(2)}(0)=\sigma_x^{(2)}\otimes I$ and 
$\sigma_x^{(2)}(\tau)=U^{\dagger}(\sigma_x^{(2)}\otimes I)U$ for $B=\sigma_x^{(2)}$.
From Eq.~(\ref{eq:Xtau}) we have
\begin{align*}
\sigma_x^{(1)}(0)&=\sigma_x\otimes I\otimes I,\\
\sigma_x^{(1)}(\tau)&=\sigma_x\otimes I\otimes I,\\
\sigma_x^{(2)}(0)&=I\otimes\sigma_x\otimes I,\\
\sigma_x^{(2)}(\tau)&=I\otimes \sigma_x\otimes\sigma_x.
\end{align*}
Let $\mu_t$ be the JPD of $\sigma_x^{(1)}(t)$ and $\sigma_x^{(2)}(t)$ in the state 
$|{\Phi^{+},0}\rangle=|{\Phi^{+}}\rangle\otimes|{0}\rangle$.
We shall show 

(i) $\mu_0(u,v)=\dfrac{1}{2}\delta_{u,v}$,\vspace{4pt}

(ii) $\mu_\tau(u,v)=\dfrac{1}{4}$,

(iii) $\delta_G(\mu_{\tau})=\eta_O(\sigma_x^{(2)})=\sqrt{2}$.

We have
\begin{align*}
\mu_0(u,v)&=\langle{\Phi^{+},0|P^{\sigma_x^{(1)}(0)}(u)P^{\sigma_x^{(2)}(0)}(v)|\Phi^{+},0}\rangle
=\langle{\Phi^{+}|P^{\sigma_x}(u)\otimes P^{\sigma_x}(v)|\Phi^{+}}\rangle\\
&=\frac{1}{2}\langle{0_x0_x|P^{\sigma_x}(u)\otimes P^{\sigma_x}(v)|0_x0_x}\rangle
+\frac{1}{2}\langle{1_x1_x|P^{\sigma_x}(u)\otimes P^{\sigma_x}(v)|0_x0_x}\rangle\\
&\quad
+\frac{1}{2}\langle{0_x0_x|P^{\sigma_x}(u)\otimes P^{\sigma_x}(v)|1_x1_x}\rangle
+\frac{1}{2}\langle{1_x1_x|P^{\sigma_x}(u)\otimes P^{\sigma_x}(v)|1_x1_x}\rangle\\
&=\dfrac{1}{2}\delta_{u,v},
\end{align*}
and (i) follows.

We have
\begin{align*}
\lefteqn{
P^{\sigma_x\otimes I\otimes I}(\pm 1)P^{I\otimes \sigma_x\otimes\sigma_x}(+1)
}\quad\\
&=
(P^{\sigma_x}(\pm 1)\otimes I\otimes I)
[I\otimes P^{\sigma_x}(+1)\otimes P^{\sigma_x}(+1)
+I\otimes P^{\sigma_x}(-1)\otimes P^{\sigma_x}(-1)]\\
&=
P^{\sigma_x}(\pm1)\otimes P^{\sigma_x}(+1)\otimes P^{\sigma_x}(+1)
+P^{\sigma_x}(\pm1)\otimes P^{\sigma_x}(-1)\otimes P^{\sigma_x}(-1),\\
\lefteqn{
P^{\sigma_x\otimes I\otimes I}(\pm 1)P^{I\otimes \sigma_x\otimes\sigma_x}(-1)
}\quad\\
&=
(P^{\sigma_x}(\pm 1)\otimes I\otimes I)
[I\otimes P^{\sigma_x}(+1)\otimes P^{\sigma_x}(-1) 
+ I\otimes P^{\sigma_x}(-1)\otimes P^{\sigma_x}(+1)]\\
&=
P^{\sigma_x}(\pm1)\otimes P^{\sigma_x}(+1)\otimes P^{\sigma_x}(-1)
+P^{\sigma_x}(\pm1)\otimes P^{\sigma_x}(-1)\otimes P^{\sigma_x}(+1).
\end{align*}
Consequently,
\begin{align*}{
\lefteqn{\sqrt{2}P^{\sigma_x\otimes I\otimes I}(+ 1)P^{I\otimes \sigma_x\otimes\sigma_x}(+1)|{\Phi^{+},0}\rangle}\quad\\
&=
P^{\sigma_x}(+1)\otimes P^{\sigma_x}(+1)\otimes P^{\sigma_x}(+1)|{0_x0_x0}\rangle
+P^{\sigma_x}(+1)\otimes P^{\sigma_x}(+1)\otimes P^{\sigma_x}(+1)|{1_x1_x0}\rangle\\
&\quad +P^{\sigma_x}(+1)\otimes P^{\sigma_x}(-1)\otimes P^{\sigma_x}(-1)|{0_x0_x0}\rangle
+P^{\sigma_x}(+1)\otimes P^{\sigma_x}(-1)\otimes P^{\sigma_x}(-1)|{1_x1_x0}\rangle\\
&=|{0_x0_x}\rangle\otimes P^{\sigma_x}(+ 1)|{0}\rangle.
}\end{align*}
Similarly,
\begin{align*}{
\sqrt{2}P^{\sigma_x\otimes I\otimes I}(-1)P^{I\otimes \sigma_x\otimes\sigma_x}(+1)|{\Phi^{+},0}\rangle
&=|{1_x1_x}\rangle\otimes P^{\sigma_x}(- 1)|{0}\rangle,\\
\sqrt{2}P^{\sigma_x\otimes I\otimes I}(+1)P^{I\otimes \sigma_x\otimes\sigma_x}(-1)|{\Phi^{+},0}\rangle
&=|{0_x0_x}\rangle\otimes P^{\sigma_x}(- 1)|{0}\rangle,\\
\sqrt{2}P^{\sigma_x\otimes I\otimes I}(-1)P^{I\otimes \sigma_x\otimes\sigma_x}(-1)|{\Phi^{+},0}\rangle
&=|{1_x1_x}\rangle\otimes P^{\sigma_x}(+1)|{0}\rangle.
}\end{align*}
Thus, we have
\begin{align*}
\mu_{\tau}(u,v)=\|P^{\sigma_x\otimes I\otimes I}(u)P^{I\otimes \sigma_x\otimes\sigma_x}(v)|{\Phi^{+},0}\rangle\|^2
&=\frac{1}{4}.
\end{align*}
for any $u,v=\pm 1$, and (ii) follows.  Thus, Eq.~(\ref{eq:ent-1-1}) is obtained.

From (ii), $\delta_G(\mu_{\tau})=\sqrt{2}$ follows.  From (i), $\delta_G(\mu_0)=0$,
so that it follows from Theorem \ref{th:chain-rule} (ii) that
$\eta_O(\sigma_x^{(2)})=\delta_G(\mu_{\tau})=\sqrt{2}$.  Thus, (iii) follows
and Eq.~(\ref{eq:ent-1-2}) is obtained.

\subsection{Projective $\sigma_{\theta}^{(2)}$ measurement}\label{se:ent-3}
Suppose that the observer makes 
a measurement $\mathbf{M}(\theta)$
 of $A=\sigma_{\theta}^{(2)}$ in $\mathbf{S}_1+\mathbf{S}_2$,
carried out by the measuring interaction 
$$
U=I\otimes P^{\sigma_{\theta}}(+1)\otimes I+I\otimes P^{\sigma_{\theta}}(-1)\otimes \sigma_x
$$ 
turned on from $t=0$ to
$t=\tau$ and by the subsequent measurement of the meter observable $M=\sigma_z^{(3)}$ of $\mathbf{P}=\mathbf{S}_3$ 
prepared in $|{\xi}\rangle=|{0}\rangle$.
This realizes the projective measurement of $A=\sigma_{\theta}^{(2)}$ as 
\begin{align*}
U(|{\alpha}\rangle\otimes|{0_{\theta}}\rangle\otimes|{0}\rangle)&=|{\alpha}\rangle\otimes|{0_{\theta}}\rangle\otimes|{0}\rangle,\\
U(|{\alpha}\rangle\otimes|{1_{\theta}}\rangle\otimes|{0}\rangle)&=|{\alpha}\rangle\otimes|{1_{\theta}}\rangle\otimes|{1}\rangle,
\end{align*}
for $\alpha=0,1$, where $|{0_{\theta}}\rangle:=|{\sigma_{\theta}=+1}\rangle$ and $|{1_{\theta}}\rangle:=|{\sigma_{\theta}=-1}\rangle$.
Consider the Heisenberg operators $B(0)=\sigma_x^{(2)}(0)=\sigma_x^{(2)}\otimes I$ and 
$B(\tau)=\sigma_x^{(2)}(\tau)=U^{\dagger}(\sigma_x^{(2)}\otimes I)U$ for $B=\sigma_x^{(2)}$.
We have
\begin{align*}{
\sigma_x^{(1)}(0)&=\sigma_x\otimes I\otimes I,\\
\sigma_x^{(1)}(\tau)&=\sigma_x\otimes I\otimes I,\\
\sigma_x^{(2)}(0)&=I\otimes\sigma_x\otimes I,\\
\sigma_x^{(2)}(\tau)&=I\otimes\sin\theta\sigma_{\theta}\otimes I+I\otimes\cos\theta\sigma_{-\theta}\otimes\sigma_x.
}\end{align*}

Let $\mu_t$ for $t=0,\tau$ be
the JPD of $\sigma_x^{(1)}(t)$ and $\sigma_x^{(2)}(t)$ in the state 
$|{\Phi^{+},0}\rangle$, i.e., 
\begin{align*}
\mu(u,v)&=\langle{\Phi^{+},0|P^{\sigma_x^{(1)}(t)}(u)P^{\sigma_x^{(2)}(t)}(v)|\Phi^{+},0}\rangle.
\end{align*}
Then we have 
$$
\mu_{0}(u,v)=\frac{1}{2}\delta_{u,v}.
$$

For $u=\pm 1$ we have
\begin{align*}{
\lefteqn{\|P^{\sigma_x^{(2)}(\tau)}(u)|{\Phi^{+},0}\rangle\|^2}\quad\\
&=\|P^{\sigma_x^{(2)}(0)}(u)U|{\Phi^{+},0}\rangle\|^2\\
&=\frac{1}{2}(\|P^{\sigma_x}(u)P^{\sigma_{\theta}}(+1)|{0}\rangle\|^2+\|P^{\sigma_x}(u)P^{\sigma_{\theta}}(-1)|{0}
\rangle\|^2
+\|P^{\sigma_x}(u)P^{\sigma_{\theta}}(+1)|{1}\rangle\|^2\\
&\quad+\|P^{\sigma_x}(u)P^{\sigma_{\theta}}(-1)|{1}\rangle\|^2)\\
&=\frac{1}{4}(\|P^{\sigma_x}(u)|{0}\rangle\|^2+\|P^{\sigma_x}(u)\sigma_{\theta}|{0}\rangle\|^2
+\|P^{\sigma_x}(u)|{1}\rangle\|^2+\|P^{\sigma_x}(u)\sigma_{\theta}|{1}\rangle\|^2)\\
&=\frac{1}{4}(\|P^{\sigma_x}(u)|{0_x}\rangle\|^2
+\|P^{\sigma_x}(u)|{1_x}\rangle\|^2
+\|P^{\sigma_x}(u)\sigma_{\theta}|{0_x}\rangle\|^2
+\|P^{\sigma_x}(u)\sigma_{\theta}|{1_x}\rangle\|^2)\\
&=\frac{1}{2}.
}\end{align*}
We used the parallelogram law twice in the third last and the penultimate equalities.
It follows that
\begin{align*}{
\sum_{u}\, \mu_{\tau}(u,v)=\frac{1}{2}.
}\end{align*}
Thus, 
with $\sum_{u}\, \mu_{0}(u,v)=1/2$,
the projective measurement of $\sigma_{\theta}^{(2)}$ is
distributionally non-disturbing to $\sigma_x^{(2)}$.

Let $\delta_{G}(\mu_t)$ be the classical root-mean-square deviation
for $\mu_t$.  We have $\delta_{G}(\mu_0)=0$.
By Theorem \ref{th:chain-rule} (ii) we have
$
\delta_{G}(\mu_\tau)=\eta_O(\sigma_x^{(2)}).
$
Then from (Ref.~\cite{12EDU}, Eq.~(6)) we have
\begin{align*}{
\eta_O(\sigma_x^{(2)})^2
=\sum_{y}\|[P^{\sigma_{\theta}^{(2)}}(v),\sigma_x^{(2)}]|{\Phi^{+}}\rangle\|^2
=2\|[\sigma_{\theta}^{(2)}/2,\sigma_x^{(2)}]|{\Phi^{+}}\rangle\|^2
=2\cos^2\theta.
} \end{align*}
Thus, we obtain Eq.~(\ref{eq:entanglement-4}), i.e.,
\begin{align*}{
\delta_{G}(\mu_\tau)=\eta_O(\sigma_x^{(2)})=\sqrt{2}\cos\theta.
}\end{align*}

In what follows we will determine $\mu_\tau$
without tedious calculations on relevant projections. 
We have
\begin{align*}{
\sum_{u,v:u\ne v}\mu_{\tau}(u,v)
=\frac{1}{4}\delta_{G}(\mu_{\tau})
=\frac{1}{2}\cos^2\theta.
}\end{align*}
Since $\sigma_x^{(1)}(0)=\sigma_x^{(1)}(\tau)$ and  
$\mathbf{M}(\theta)$ is distributionally non-disturbing to
$\sigma_x^{(2)}$, we have
\begin{align*}{
\sum_{v}\mu_{\tau}(u,v)&=\sum_{v}\mu_{0}(u,v)=\frac{1}{2},\\
\sum_{u}\mu_{\tau}(u,v)&=\sum_{u}\mu_{0}(u,v)=\frac{1}{2}.
}\end{align*}
Since, 
\begin{align*}{
\sum_{u,v:u\ne v}\mu_{\tau}(u,v)+\sum_{v}\mu_{\tau}(+1,v)
&=\mu_{\tau}(+1,-1)\!+\!\mu_{\tau}(-1,+1)\!+\!\mu_{\tau}(+1,-1)\!
+\!\mu_{\tau}(+1,+1)\quad\\ 
&=2\mu_{\tau}(+1,-1)+\sum_{u}\mu_{\tau}(u,+1), 
}\end{align*}
we obtain
\begin{align*}{
\mu_{\tau}(+1,-1)=\frac{1}{2}\sum_{u,v:u\ne v}\mu_{\tau}(u,v)=\frac{1}{4}\cos^2\theta.
}\end{align*}
It follows that 
\begin{align*}{
\mu_{\tau}(+1,-1)&=\mu_{\tau}(-1,+1)=\frac{1}{4}\cos^2\theta,\\
\mu_{\tau}(+1,+1)&=\mu_{\tau}(-1,-1)=\frac{1}{4}(1+\sin^2\theta).
}\end{align*}
Therefore, we have derived Eq.~(\ref{eq:entanglement-3}).
\bigskip

\renewcommand{\&}{and}
\setlength{\itemsep}{0in}

\end{document}